\documentclass[pra,aps,twocolumn,superscriptaddress,longbibliography]{revtex4-1}
\usepackage[colorlinks=true, citecolor=red, urlcolor=blue ]{hyperref}
\usepackage{graphicx}
\usepackage{bm}
\usepackage{amsmath,amsfonts}
\usepackage{amsthm}
\usepackage{hyperref}
\usepackage{xcolor}
\hypersetup{
    urlcolor=blue,
    citecolor=blue
           }

\usepackage{derivative}
\usepackage{braket}
\theoremstyle{plain}
\usepackage{color}
\usepackage{amssymb}
\usepackage{amsthm}
\usepackage{amsfonts}
\usepackage{float}
\usepackage{tabularx}
\usepackage{graphicx}
\usepackage[utf8]{inputenc}

\usepackage{mathtools}
\usepackage{esvect}
\usepackage{wrapfig}
\usepackage{amsthm}
\usepackage{verbatim}
\usepackage{bbm}
\usepackage[normalem]{ulem}

\usepackage{enumitem}
\usepackage{fmtcount}
\usepackage{booktabs}
\usepackage{csquotes}
\usepackage{epsfig}

\usepackage{tabularx}
\usepackage{graphicx}
\usepackage{amsmath}
\usepackage{braket}
\usepackage{latexsym}
\usepackage{bm}
\usepackage{graphics,epstopdf}
\usepackage{enumitem}
\usepackage{fmtcount}
\usepackage{booktabs}
\usepackage{csquotes}
\usepackage{epsfig}

\theoremstyle{plain}

\def\bea{\begin{eqnarray}}
\def\eea{\end{eqnarray}}
\def\ba{\begin{array}}
\def\ea{\end{array}}
\def\n{\nonumber}
\def\c{\mathscr}
\def\beq{\begin{equation}}
\def\eeq{\end{equation}}

\usepackage[normalem]{ulem}
\usepackage{float}
\usepackage{graphicx}  
\usepackage{dcolumn}          
\usepackage{amssymb}
\usepackage{appendix}
\usepackage{physics}   
\usepackage{mathtools}
\usepackage{esvect}
\usepackage{wrapfig}
\usepackage{amsthm}
\usepackage{verbatim}
\usepackage{bbm}
\usepackage{soul}

\usepackage[mathscr]{euscript}
\def\Tr{\operatorname{Tr}}
\def\Pr{\operatorname{Pr}}
\def\sq{\operatorname{sq}}
\def\SEP{\operatorname{SEP}}
\def\FS{\operatorname{FS}}
\def\GE{\operatorname{GE}}
\def\E{\operatorname{E}}
\def\BS{\operatorname{BS}}
\def\Ent{\operatorname{Ent}}
\def\PPT{\operatorname{PPT}}
\def\supp{\operatorname{supp}}
\def\LOCC{\operatorname{LOCC}}
\def\GHZ{\operatorname{GHZ}}
\def\T{\operatorname{T}}

\def\id{\operatorname{id}}

\def\({\left(}
\def\){\right)}
\def\[{\left[}
\def\]{\right]}

\let\oldemptyset\emptyset
\let\emptyset\varnothing
\def\bound{\boldsymbol{\mathsf{L}}}
\newcommand{\mc}[1]{\mathcal{#1}}
\newcommand{\wt}[1]{\widetilde{#1}}
\newcommand{\wh}[1]{\widehat{#1}}
\newcommand{\tf}[1]{\textbf{#1}}
\newcommand{\tn}[1]{\textnormal{#1}}
\newcommand{\msc}[1]{\mathscr{#1}}
\newcommand{\bbm}[1]{\mathbbm{#1}}

\renewcommand{\qedsymbol}{$\blacksquare$}

\renewcommand{\t}{{\scriptscriptstyle\mathsf{T}}}
\DeclarePairedDelimiter{\ceil}{\lceil}{\rceil}
\newtheorem{theorem}{Theorem}

\newtheorem{remark}{Remark}

\renewcommand{\arraystretch}{1.3}

\begin{document}

\title{Non-positive measurements aren't beneficial in quantum metrology for unitary encoding, but can be for open schemes}





\author{Paranjoy Chaki}
\affiliation{Harish-Chandra Research Institute,  A CI of Homi Bhabha National Institute, Chhatnag Road, Jhunsi, Allahabad 211 019, India}

\author{Debarupa Saha}
\affiliation{Harish-Chandra Research Institute,  A CI of Homi Bhabha National Institute, Chhatnag Road, Jhunsi, Allahabad 211 019, India}

\author{Kornikar Sen}
\affiliation{Departamento de Física Teórica, Universidad Complutense, 28040 Madrid, Spain}

\author{Ujjwal Sen}
\affiliation{Harish-Chandra Research Institute,  A CI of Homi Bhabha National Institute, Chhatnag Road, Jhunsi, Allahabad 211 019, India}

\begin{abstract}
We investigate whether {non-positive operator-valued measurements} can be beneficial for quantum metrology.
For unitary encoding, we 
show that non-positive measurements offer no advantage over positive ones. 
Going over to open encoding, we find, however, that non-positive measurements can be advantageous for certain cases, while it may mirror the unitary case - no advantage over positive measurements - for others. 
For arbitrary open-system encoding, we identify a sufficient condition under which positive measurements suffice to achieve the best precision, 
and 
more resource-intensive non-positive measurements offer no extra benefit.

\end{abstract}
\maketitle

\section{Introduction}

Quantum metrology~\cite{QM,QM2}, the science of estimation, plays a foundational role in quantum information theory. It ensures that physical quantities intrinsic to a quantum system are measured accurately, which in turn is essential for quantum communication~\cite{qq_c,qq_c11,com_m1,com_m2}, quantum cryptography~\cite{cryp_1,cryp_2}, and sensing~\cite{QS,QS2}. By carefully selecting the measurement setting, one can minimize the error in parameter estimation and approach the theoretical lower bound, known as the {Cramér–Rao bound}~\cite{CM1,CM2,CM3}. The sensitivity of initial probe systems in parameter estimation is analyzed in Refs.~\cite{prob_1,prob_3,prob_5}. Precision in estimation can often be enhanced by exploiting quantum properties such as entanglement~\cite{Entx}, coherence~\cite{Coh}, non-Markovianity~\cite{non-Mar}, and squeezing~\cite{q_squ}, as demonstrated in Refs.~\cite{Ep5,Ep6,Ep7,Ep8,Ep9,En1,En2,En3,En4,En5,En6,NM1,NM2,NNMMf,NNMM1,NNMM2,sup_1,sq1,sq2,squ1}. These properties help enhance precision beyond classical limits, referred to as the standard quantum limit, ultimately reaching the Heisenberg limit~\cite{HL1,HL11,HL2,HL3,HL4}.  Apart from this, the application of quantum metrology in many-body systems is explored in Refs.~\cite{E1,E3,E14,E5,E4,E6,E7,E8}.


In quantum metrology, there are two types of parameter encoding processes, namely unitary encoding~\cite{prob_2} and open encoding~\cite{open_X1,open_x2,open_x3,open_x4,open_x5,open_1,suva_2}. After the parameter is encoded into the probe system, a measurement is typically performed to estimate the parameter. The measurement essentially {acts to decode} the information about the encoded parameter. Various measurement strategies have been employed for parameter estimation in the context of quantum metrology, such as imperfect and weak measurements~\cite{un_sharp,weak_1,weak_2,weak_3}, sequential measurements~\cite{seq_3,seq_2,seq_1}, incompatible measurements~\cite{inc_measure1}, local measurements~\cite{loc_measurement}, and random measurements~\cite{measurement_random}. To the best of our knowledge, all these types of measurement schemes fall under the category of positive operator-valued measurements (POVMs). A standard method for implementing a POVM on a system involves introducing an auxiliary system that is initially uncorrelated with the probe system, performing a global projective measurement on the combined system and auxiliary, and then tracing out the auxiliary system. However, the measurement setting can be generalized by considering the initial probe–auxiliary state that are possibly correlated. By correlating the system and the auxiliary, one can realize a new type of measurement setting beyond POVMs. Such measurements are known as non-positive operator-valued measurements (NPOVMs)~\cite{N2,N3}. The non-positivity of quantum measurements and quantum channels has been shown to provide advantages in several areas of quantum technologies, one prominent example being energy-extraction protocols for quantum batteries~\cite{N1,N2}. 

  Motivated by the advantages of non-positiveness in measurements over POVM ones~\cite{N2,N3}, we raise the natural question: can non-positivity in measurement strategies offer enhanced precision in quantum metrology? Addressing this question not only extends the theoretical framework of quantum measurements but also highlights the potential of NPOVMs to provide a new advantage in quantum metrology, thereby broadening their relevance within quantum technologies.

To explore this, we consider both unitary and open encoding strategies in quantum metrology. In the case of arbitrary unitary encoding processes, we analytically show that standard POVMs are sufficient to achieve the optimal precision, even within the broader class of measurement settings, general measurements that include both positive and non-positive measurement (decoding) strategies. Here the optimization is performed over the initial state and measurement settings.This ensures that, for arbitrary unitary encodings, non-positiveness in measurement strategies provides no extra advantage.

However, in the case of open encoding, we encounter two distinct scenarios. In one scenario, positive measurement and general measurement schemes yield the same optimal precision; in the other, NPOVM measurement schemes outperform POVM reflecting better performance than the measurement direction corresponding to the eigenbasis of the symmetric logarithmic derivative operator. To further clarify the distinction between these two cases, we provide a sufficient condition that determines when, for a given open encoding process realized by a two-qubit global unitary and a single-qubit environment, positive measurement strategies suffice to achieve the same precision for a suitable choice of initial probe system as that achieved by any given general measurement for any given initial probe auxiliary system. Here, the probe system is also considered to be a single qubit system. From here it also follows that if we know the optimal initial state of the probe and auxiliary system corresponding to the general measurement, then this condition also serves as the sufficient condition for achieving the best precision by POVM over the general measurements. We also provide some examples, such as estimating the noise strengths of bit-flip and dephasing channels and estimating an arbitrary parameter encoded via $XX$ interactions between the environment and the auxiliary system.

Conversely, we also present explicit, physically motivated examples where the benefits
of a non-positive measurement scheme is depicted. In particular, we consider the transverse-field XY (TXY) model and analyze three scenarios in which different parameters of the model are estimated, namely, field strength, interaction strength, and anisotropy parameter. In each of these cases, NPOVMs are shown to provide better precision compared to standard POVMs.

The rest of the paper is organized as follows. In Sec.~\ref{premea}, we provide a brief discussion on both positive and general quantum measurement settings for non-positive measurements. Sec.~\ref{para} contains a prerequisite discussion on the processes underlying quantum metrology. In Sec.~\ref{3}, we present the case of unitary encoding and show that positive measurement strategies suffice to achieve the optimal precision. In Sec.~\ref{4}, we turn to the case of open encoding and derive a sufficient condition that characterizes when POVM measurements remain optimal. This is followed by numerical examples, which support the derived condition. In Sec.~\ref{4Bx}, we present examples demonstrating that when the sufficient condition is not satisfied, one can unlock the advantage of NPOVM measurement settings in quantum metrology. And lastly, we conclude in Sec.~\ref{5}.


\section{preliminaries}\label{2}

In this section, we present a brief overview of POVMs, physically realizable NPOVM operations, and quantum parameter estimation theory under both POVM and NPOVM decoding strategies. The discussion on POVM and NPOVM measurement settings is presented below.

\subsection{Brief discussion on POVMs and general quantum measurements}
\label{premea}

Consider a system $S$ of dimension $d_S$ and an auxiliary $A$ of dimension $d_A$. Suppose the initial joint state of $S$ and $A$ is a product state, i.e., $\rho_S \otimes \rho_A$, where $\rho_S$ and $\rho_A$ denote the states of the system and the environment, defined in Hilbert spaces $\mathcal{H}^S$ and $\mathcal{H}^A$, respectively. If a projective measurement is performed on $\rho_S \otimes \rho_A$ using an element $Q_i$ from a set of orthogonal projectors $\{Q_i\}$, where each element of the set acts on the joint Hilbert space $\mathcal{H}^S \otimes \mathcal{H}^A$ and satisfies $Q_i Q_j = \delta_{ij}$ for all $i,j$, and $\sum_i Q_i = \mathbb{I}_d$ with $d = d_S d_A$. Note here and throughout the paper, we use the notation $\mathbb{I}_d$ to denote the identity operator on the $d$-dimensional space. The reduced state of the system corresponding to each projection outcome $i$ in such a scenario is given as
$$
\rho_{S}^i = \frac{\mathrm{Tr}_A\bigl[Q_i \,(\rho_S \otimes \rho_A)\, Q_i\bigr]}{\mathrm{Tr}\bigl[Q_i \,(\rho_S \otimes \rho_A)\bigr]}=\frac{\chi_i \rho_{S} \chi_i^{\dag}}{\tr[ \rho_{S} \chi_i^{\dag}\chi_i]}.
$$

Here, $\chi_i$ denotes the effective measurement operator acting on the system state $\rho_{S}$, and $\tr_{A}$ refers to the partial trace over $A$. Such a measurement, implemented on the subsystem $S$ by performing a projective measurement on the composite system $\rho_S \otimes \rho_A$, is referred to as a \emph{positive operator-valued measurement} (POVM). The term ``positive'' in POVM refers to the fact that the probability $p_i$ of obtaining a specific outcome $\rho^i_{S}$ can be written as
$ p_i = \tr[\rho_{S} \chi_i^{\dag} \chi_i] = \tr[\rho_{S} E_i], $
where each $E_i = \chi_i^{\dag} \chi_i$ is positive semidefinite and is referred to as a POVM element corresponding to the outcome labeled by $i$.  

The properties that any set of POVM elements $\{E_i\}$ must satisfy are: \\ 
1. Each $E_i$ is Hermitian, i.e., $E_i = E_i^{\dag}$, $\forall i$,  \\
2. The eigenvalues of each $E_i$ are non-negative,  \\
3. The completeness relation holds, i.e., $\sum_i E_i = \mathbb{I}_{d_S}$.  

Such a POVM measurement on the system $S$ reduces to a projective measurement when the projective measurement $\{Q_i\}$ performed on the composite system $\rho_S \otimes \rho_A$ forms a product basis. In this case, the POVM elements $\{E_i\}$ additionally satisfy  
$ E_i E_j = \delta_{ij} E_i, \quad \forall i,j. $
For projective measurement the auxiliary system $B$ becomes irrelevant.

Note that the very idea of POVM measurements relies on the implicit assumption that the initial composite state of the system and the auxiliary is a product state. However, one can also prepare initial composite states $\rho_{SA}$ of the system and the auxiliary that may be correlated, and then perform a projective measurement $\{Q_i\}$ on the joint system. In such a scenario, the final state of the system $S$ corresponding to a measurement outcome labeled ``i" becomes.
\begin{equation*}
\tilde{\rho}^{i}_{S} = \frac{\Tr_{A}\left[Q_i \rho_{SA} Q_i\right]}{\Tr\left[Q_i \rho_{SA}\right]}.
\end{equation*}

Such a measurement strategy is called a \emph{general quantum measurement}, where, unlike in the POVM case, it is often not possible to use positive semidefinite operators of the form $E_i$ to indicate the likelihood of obtaining a specific outcome $i$, mainly because the initial state $\rho_{SA}$ may be correlated. Nevertheless, one can still regard general quantum measurements as effective measurements on $S$ that arise from performing a projective measurement on the (possibly correlated) composite state $\rho_{SA}$ and then discarding the auxiliary. It is important to note that general quantum measurements form a superset of POVM measurements, i.e., they include both POVM measurements as well as those for which effective positive semidefinite operators of the form $E_i$ cannot be obtained, and are therefore referred to as non-positive operator-valued measurements (NPOVMs).
 

Having discussed the basics of both POVM and general quantum measurement setting we now move on to the topic of quantum parameter estimation and discuss how the aforementioned two measurement schemes can be employed to infer the parameter. The discussion is presented below.

\subsection{Quantum parameter estimation}
\label{para}
Suppose we want to estimate a parameter, $\theta$, which could be a component of the system's Hamiltonian~\cite{E3,E5}, the transition frequency of atomic clocks~\cite{atomic,atom_ap}, or a quantity defining a quantum channel~\cite{ch_11,ch_1}. Irrespective of the nature of $\theta$, its estimation in quantum metrology is typically performed through a two-step process. First, the parameter $\theta$ is encoded onto a probe, $\rho$ which represents the density matrix of a quantum system of dimension $d_S$, via a physical process which are either through a unitary or an open-system encoding method.  The choice of encoding depends on the nature of the parameter to be estimated. After the encoding process the encoded state of the probe becomes $\rho_{\theta}$.

In the second step, a measurement is performed on the encoded probe system, and $\theta$ is estimated based on the measurement outcomes using a suitable estimator. The commonly used estimator is an unbiased estimator, $\mathcal{K}(i)$, that satisfies the condition.
\begin{equation}
    \langle\mathcal{K}(i)\rangle_\theta = \int di \, p(i|\theta)\mathcal{K}(i) = \theta.
\end{equation}
Here, $p(i|\theta)$ denotes the probability of obtaining a particular outcome labeled $i$ when measuring the encoded state $\rho_{\theta}$. For a given parameter $\theta$, an estimator is considered unbiased if its average over the measurement outcomes $i$ equals the true value of $\theta$.

The primary objective is to estimate the parameter with the smallest possible error.


The spread (standard deviation) in estimating the parameter $\theta$ reflects the estimation error, thus the goal is to minimize this error by selecting the optimal estimator, probe and measurement, that gives the smallest possible standard deviation. For a particular probe state and measurement setting, there is a fundamental lower bound on the standard deviation of the parameter $\theta$ obtained by optimizing over all possible unbiased estimators. For single-shot measurement such a bound is known as the \textbf{Cramér-Rao bound}~\cite{CM1,CM2,CM3}, given by:
\begin{equation}
   \Delta\theta\geq\frac{1}{\sqrt{\mathcal{F}_{\rho_S}^{U/O}(\theta)}}.
\end{equation}
Here, $\rho_S$ in the suffix represent the initial input state and $\mathcal{F}_{\rho_S}(\theta)$ is the Fisher information corresponding to the measurement outcome $p(i|\theta)$ extracted from the encoded probe state $\rho_{S}(\theta)$ by performing  measurement, and is given by

\begin{equation}\label{3p}
\mathcal{F}_{\rho_S}(\theta)=\int di~p(i|\theta)  \left[\pdv{\log( p(i|\theta))}{\theta}\right]^2.
\end{equation}
\begin{figure*}
		\centering
            \includegraphics[width=8cm]{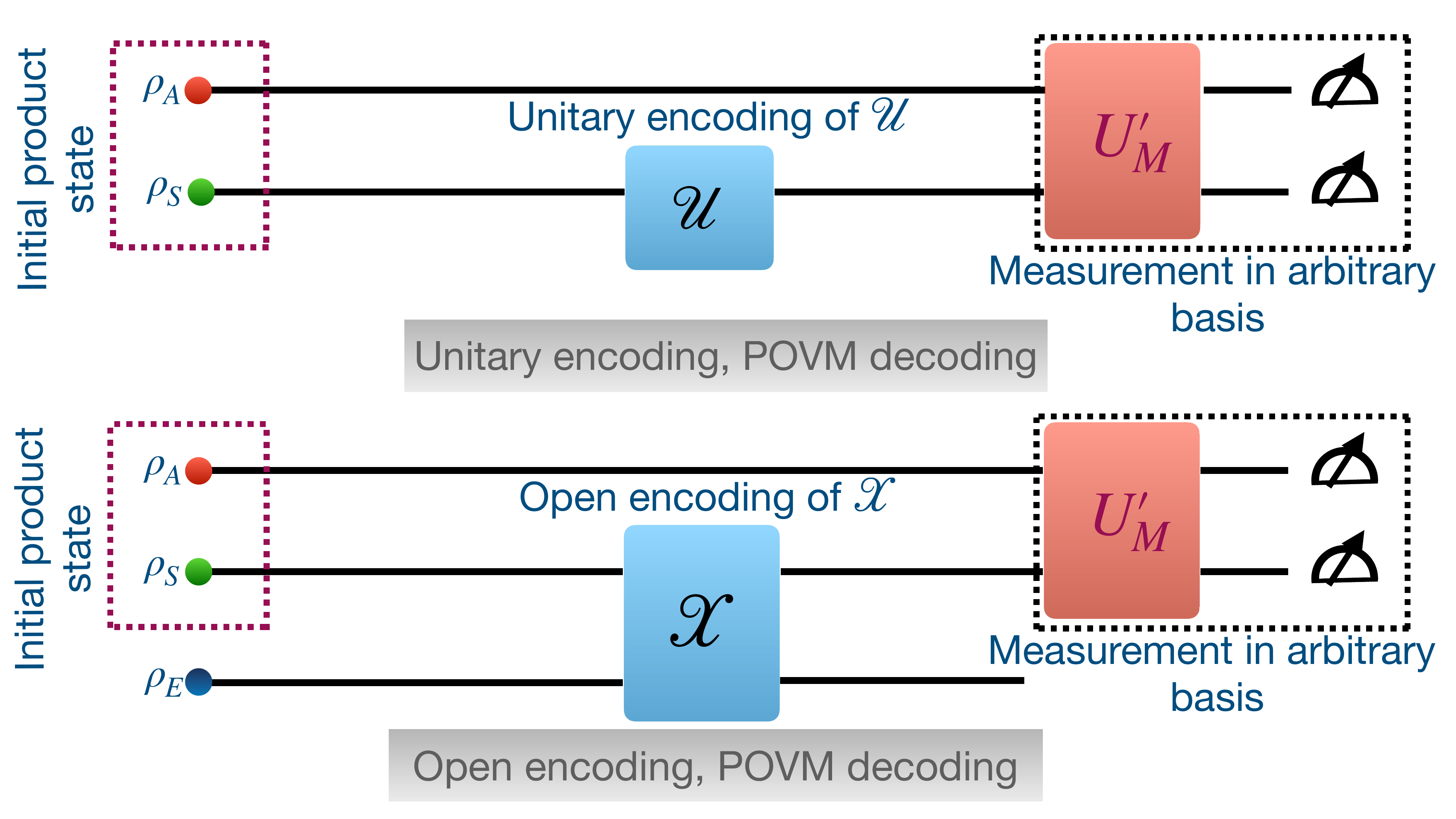}
            \hspace{1cm}
            \includegraphics[width=8cm]{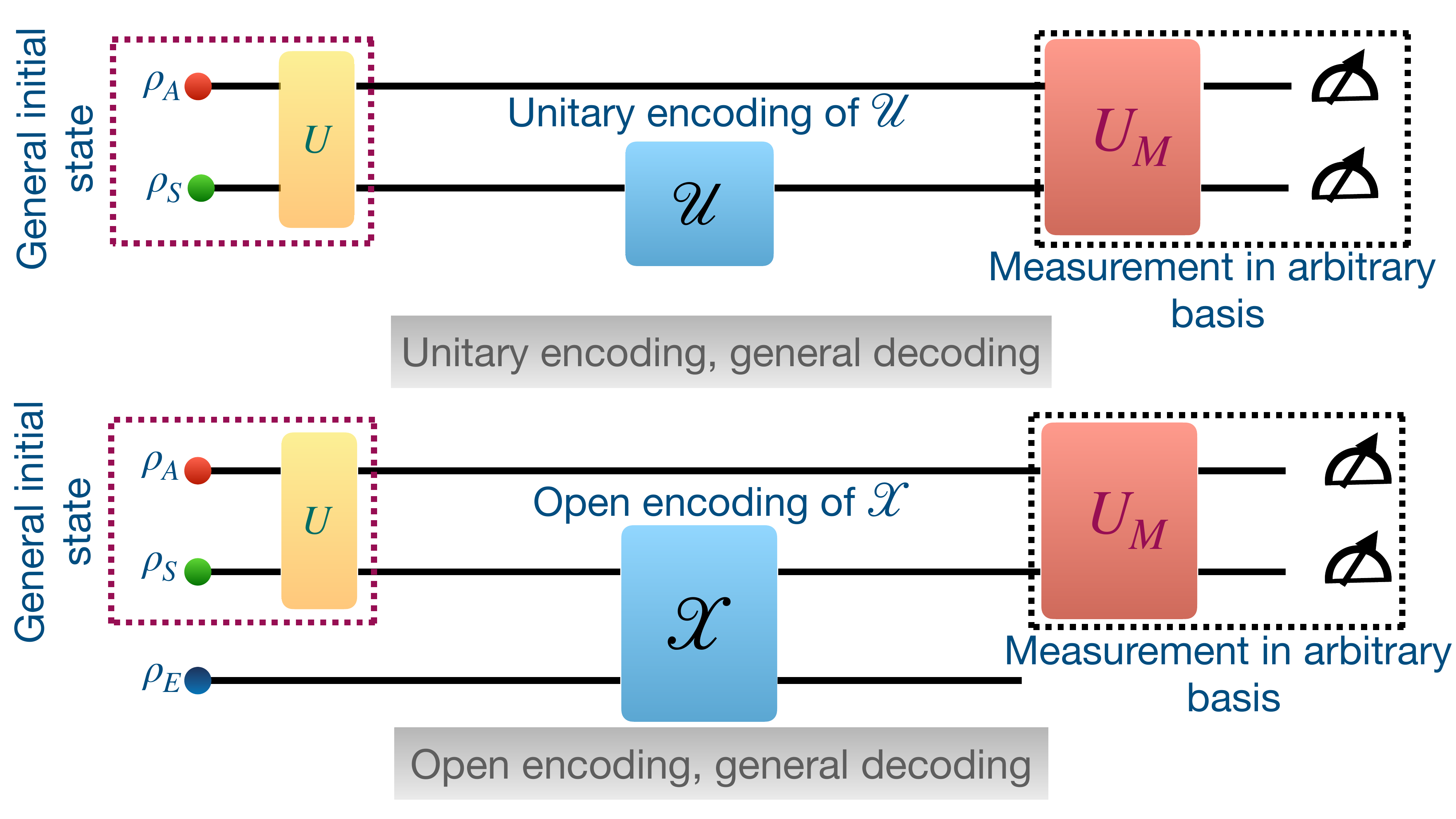}
\caption{{\textbf{POVMs and general measurements in quantum metrology.}} The figure depicts a schematic of POVMs and generalized quantum operations for parameter estimation under both unitary and open encoding. The probe $\rho_S$, auxiliary $\rho_A$, and environment $\rho_E$ are represented by green, red, and navy blue dots, respectively, while unitary (both local and global) used for encoding and measurement operations are indicated by yellow, blue, and red boxes. The left panel shows POVM decoding: the upper half corresponds to unitary encoding with $\mathcal{U}$, and the lower half to open encoding with $\mathcal{X}$, where the probe $\rho_S$ and environment $\rho_E$ are initially uncorrelated, and $\mathcal{X}$ involves the application of a global unitary. In both cases, a projective measurement is applied on the probe and auxiliary, with the global unitary $U_M'$ and $U_M$ fixing the measurement basis corresponding to POVMs and general measurements. The right panel shows the general decoding strategy. Here, an initial global unitary $U$ might correlate $\rho_A$ and $\rho_S$ before encoding by $\mathcal{U}$ or $\mathcal{X}$. The key distinction between the left and right panels is that in the general case the auxiliary and probe may start correlated, and optimizing the QFI over all general measurements requires optimization over  $\rho_A$, $U$, and $U_M$.  
}
		\label{fig3x} 
		\end{figure*} 

As seen from Eq.~\eqref{3p}, the Fisher information $\mathcal{F}(\theta)$ is the one obtained by optimizing over all possible unbiased estimator and therefore it does not depend on the estimator function $\mathcal{K}(i)$. For a particular initial state $\rho_S$, this Fisher information can further be optimized over measurement settings. Usually optimization over POVM measurement setting, is considered for this purpose. We refer to such a measurement strategy as the POVM decoding. Such a maximization gives the quantum Fisher information defined as
\begin{equation*}
    \mathcal{F}_{Q|\rho_S}(\theta) = \max_{X \in \Lambda^P} \mathcal{F}_{\rho_S}(\theta).
\end{equation*}
Here, ``$Q|\rho_S$'' in the subscript of $\mathcal{F}_{Q|\rho_S}(\theta)$ denotes the quantum Fisher information corresponding to the initial state $\rho_S$. $X$ is a particular choice of measurement setting and $\Lambda^P$ denotes the set of all possible POVM measurements defined in the Hilbert space $\mathcal{H}^d$.
Thus we have $ \mathcal{F}_{Q|\rho_S}(\theta) \geq \mathcal{F}_{\rho_S}(\theta)$.
 
Note that the discussion so far holds for any type of encoding, including both unitary and open encoding. However, to distinguish between the types of encoding considered, we denote the Fisher information $\mathcal{F}_{\rho_S}(\theta)$ as $\mathcal{F}_{\rho_S}^{U/O}(\theta)$, and 
the QFI $\mathcal{F}_{Q|\rho_S}(\theta)$ as $\mathcal{F}_{Q|\rho_S}^{U/O}(\theta)$, where $U$ denotes unitary encoding and $O$ denotes open encoding.

For a particular type of encoding the QFI can be obtained by measuring the encoded probe, $\rho_S{(\theta)}$, in the eigenbasis of a special operator, known as the symmetric logarithmic derivative (SLD) operator. Let us consider $\mathcal{L}$ is the SLD operator corresponding to a state $\rho_S{(\theta)}$. Then the following equation holds,
\begin{equation}
    \pdv{\rho_S{(\theta)}}{\theta}=\frac{1}{2}[\mathcal{L}\rho_S{(\theta)}+\rho_S{(\theta)}\mathcal{L}].
\end{equation}
The quantum fisher information information(QFI) in terms of encoded probe state $\rho_S{(\theta)}$ and SLD operator is given by,
\begin{equation}
  \mathcal{F}_{Q|\rho_S}^{U/O}(\theta)=\Tr{[\rho_S{(\theta)}\mathcal{L}^2]}\nonumber.
\end{equation}
For initial pure probe state,  $\ket{\psi_0}$ and unitary encoding of the form $ U_{\theta}=\exp(-\iota  H\theta )$. The QFI can be expressed as,
\begin{equation}
 \mathcal{F}_{Q|\rho_S}^{U}(\theta) = 4\left[ \langle \dot{\psi}_S(\theta) | \dot{\psi}_S(\theta) \rangle - |\langle \dot{\psi}_S(\theta) | \psi_S(\theta) \rangle|^2 \right]\nonumber.
\end{equation} Here, $\ket{\psi_S(\theta)}$ is the encoded state given as $\ket{\psi_S(\theta)}=U_{\theta}\ket{\psi_S}$. Then the QFI in this scenario just becomes, $\mathcal{F}_{Q|\rho_S}^{U}(\theta)=4 \Delta^2H$. With $\Delta^2H$ being the variance of the Hamiltonian $H$ with respect to the initial state $\rho_S=\ketbra{\psi_S}{\psi_S}$.

 Here we emphasize that the traditional QFI is obtained by restricting the optimization to the set of POVM measurement settings. However, the measurement settings can be further general by including NPOVM measurement settings as well. Therefore, a general quantum measurement setting should include both POVM and NPOVM measurement settings. In this work, we consider such general measurement settings and optimize $\mathcal{F}_{\rho_S}(\theta)^{U/O}$ over this broader class. Consequently, the generalized QFI for a given $\rho$ is defined as
\begin{equation}
    \mathcal{F}_{G|\rho_S}^{U/O}(\theta) := \max_{X \in \Lambda^G} \mathcal{F}_{X|\rho_S}^{U/O}(\theta),
\end{equation}
where ``$X|\rho_S$", ``$G|\rho_S$" in the subscript denotes QFI for a particular measurement setting X and generalized QFI, respectively for a given initial state $\rho_S$. $\Lambda^G$ denotes the set of all possible measurement settings, i.e., POVM $\cup$ NPOVM. Thus in general we have $\mathcal{F}_{G|\rho_S}^{U/O}(\theta)\geq \mathcal{F}_{Q|\rho_S}^{U/O}(\theta)$, and hence the {Cramér-Rao bound} can now be written as  
\begin{equation*}
     \Delta\theta \geq\frac{1}{\sqrt{\mathcal{F}_{\rho_S}^{U/O}(\theta)}}\geq\frac{1}{\sqrt{\mathcal{F}_{Q|\rho_S}^{U/O}(\theta)}}\geq\frac{1}{\sqrt{\mathcal{F}_{G|\rho_S}^{U/O}(\theta)}}.
\end{equation*}
Note that one can further optimize both the QFI and the generalized QFI over all possible choices of the initial probe state $\rho_S$. This yields the optimal QFI, $\mathcal{F}_Q^{U/O}(\theta)$, and the optimal generalized QFI, $\mathcal{F}_G^{U/O}(\theta)$, defined as
\begin{equation}
    \mathcal{F}_{Q}^{U/O}(\theta) := \max_{X \in \Lambda^P,\, \rho_S \in \chi_{d_S}} \mathcal{F}_{X|\rho_S}^{U/O}(\theta),
    \label{ofi}
\end{equation}
\begin{equation}
    \mathcal{F}_{G}^{U/O}(\theta) := \max_{X \in \Lambda^G,\, \rho_S \in \chi_{d_S}} \mathcal{F}_{X|\rho_S}^{U/O}(\theta).
    \label{ogfi}
\end{equation}
where $\chi_{d_S}$ denotes the set of all quantum states defined on the Hilbert space $\mathcal{H}^d$. Naturally, one then has $\mathcal{F}_{G}^{U/O}(\theta) \geq \mathcal{F}_{Q}^{U/O}(\theta)$.

In Fig.~\ref{fig3x}, we provide a schematic illustrating the implementation of POVMs and general quantum operations for parameter estimation in both types of encoding: unitary and open. In the figure, the initial probe state $\rho_S$ is represented by a green dot, the auxiliary system $\rho_A $ by a red dot, the yellow boxes denote unitary operations, the blue boxes denote encoding operations, and the red boxes represent measurements, whereas the one with the navy blue dot represents the environment $\rho_E$ considered for the open encoding. The left panel of Fig.~\ref{fig3x} depicts the case of POVM decoding. Here, the states $\rho_S$ and $\rho_E$ are initially uncorrelated. In the upper half, $\mathcal{U}$ denotes the encoding unitary, while in the lower half, $\mathcal{X}$ represents the open encoding, which intrinsically involves a global unitary operation on the probe $\rho_S$ and the environment $\rho_E$, followed by tracing out the environment. In both cases, the decoding stage applies a projective measurement on the joint state of the auxiliary and the probe, with the global unitary $U_M$ determining the measurement basis. Therefore, to optimize the QFI, one must optimize over all choices of $\rho_A$, $\rho_S$, and $U_M$. The right panel demonstrates the general quantum decoding strategy, with the upper half corresponding to unitary encoding and the lower half to open encoding. Here, a global unitary $U$ applied at the initial stage entangles $\rho_A$ and $\rho_S$. The encoding operations $\mathcal{U}$ and $\mathcal{X}$ serve the same role as in the left panel. At the decoding stage, $U_M$  determines the choice measurement basis. The essential difference between the POVM decoding (left panel) and the general decoding strategy (right panel) is the initial global unitary $U$ in the latter, which allows the possibility of the auxiliary and the probe to be initially correlated. Consequently, to optimize the generalized QFI, one must optimize over the choices of $\rho_A$, $\rho_S$, $U$, and $U_M$.

Having discussed both measurement strategies for parameter estimation, we now move on to analyze the two types of encoding strategies separately. In particular, for each encoding strategy, we investigate whether the general measurement strategy yields more Fisher information than the POVM strategy. Such a result would imply that non-positivity in the measurement strategy can be beneficial in the context of parameter estimation. For the unitary encoding strategy, however, we show that this is not the case, and the POVM strategy suffices to yield the maximum possible QFI. A detailed discussion is presented below.  

 \section{POVM decoding sufficient to attain maximum precision for unitary encoding}\label{3}

In this section, we investigate whether there are any advantages of using non-positive measurement strategies in the realm of parameter estimation for {arbitrary} unitary encoding processes. In that regard we provide our first theorem, that proves that non-positive measurement strategies offers no benefit over the POVM ones for arbitrary unitary encoding strategies. The theorem is as follows.

\begin{theorem}
For any unitary encoding process $U$, the optimal QFI $\mathcal{F}_Q^U(\theta)$, optimized over all POVM measurement strategies and input probe state, is equal to the optimal generalized QFI $\mathcal{F}_G^U(\theta)$, optimized over all general quantum measurement strategies, and input probe states. In other words, 
\begin{equation*}
    \mathcal{F}_Q^U(\theta) = \mathcal{F}_G^U(\theta).
\end{equation*}
Hence, for unitary encodings, non-POVM measurement strategies provide no advantage over POVMs.
\end{theorem}

\begin{proof}
Let the initial input system probe $\rho_S$, of dimension $d_S$, evolve under the action of a unitary of the form 
$U=\exp(-\iota H(\theta))$, 
where $H(\theta)$ is an arbitrary Hermitian operator that depends on the parameter $\theta$ to be estimated. 
The eigenvalues of $H(\theta)$ are assumed to be dimensionless. 
After the evolution, the encoded state is given by 
$\rho_S(\theta)=\exp(-\iota H(\theta))\,\rho_S\exp(\iota H(\theta))$. 
The goal is to obtain the optimal QFI, $\mathcal{F}_Q^U(\theta)$, and the generalized QFI, $\mathcal{F}_G^U(\theta)$, optimized over all choices of measurement settings in $\Lambda^P$ and $\Lambda^G$, respectively, as well as over all possible initial states $\rho_S$.

We first consider the case of obtaining $\mathcal{F}_Q^U(\theta)$. 
The convexity property of QFI~\cite{con_1,con_2,CQFI,con_3} implies that, among all possible input states, pure states maximize the QFI. 
Therefore, $\mathcal{F}_Q^U(\theta)$ can be redefined as
\begin{equation}
    \mathcal{F}_Q^U(\theta) = \max_{X \in \Lambda^G,\, \rho_S \in \zeta_d} \mathcal{F}_{X|\rho_S}^{U}(\theta),
\end{equation}
where $\zeta_{d_S} \subset \chi_{d_S}$ denotes the set of all pure states of the form $\rho_S = \ketbra{\psi_S}{\psi_S}$ of dimension $d_S$.

Thus, considering only pure entangled states $\rho_S = \ketbra{\psi_S}{\psi_S}$, the final encoded state becomes $\ket{\psi_S(\theta)} = U\ket{\psi_S} = \exp(-\iota 
H(\theta))\ket{\psi_S}$.
A POVM measurement is performed on $\ketbra{\psi_S}{\psi_S}$ to infer the parameter $\theta$. Recall that, in general, implementing a POVM requires an additional auxiliary system.
However, the QFI corresponding to the encoded state under a POVM measurement strategy can be attained by performing a projective measurement in the eigenbasis of the SLD operator, as described in Sec.~\ref{2}, and this does not depend on the auxiliary system.
Therefore, the {QFI} for the unitarily encoded state $\ket{\psi_S(\theta)}$ is given by
\begin{equation}\label{66}
    \mathcal{F}_{Q|\rho_S}^U(\theta) = 4\left[ \langle \dot{\psi_S}(\theta) | \dot{\psi_S}(\theta) \rangle - |\langle \dot{\psi_S}(\theta) | \psi_S(\theta) \rangle|^2 \right],
\end{equation}
where
\begin{eqnarray*}
\begin{split}
   \ket{\dot{\psi_S}{(\theta})}&=\frac{d\ket{\psi_S{(\theta)}}}{d\theta}=-\iota\exp(-\iota H(\theta))\,{\dot{H}(\theta)}\ket{\psi_S},\\
  {\dot{H}(\theta)}&=\frac{d{H({\theta}})}{d\theta}.
  \end{split}
\end{eqnarray*}
Substituting $\ket{\psi_S(\theta)}=\exp(-\iota H(\theta))\ket{\psi_S}$ into Eq.~\eqref{66}, we obtain
\begin{eqnarray}
\mathcal{F}_{{Q|\rho_S}}^{U}\nonumber(\theta)&=&4\left[\bra{\psi_S}{\dot{H}(\theta)}^{2}\ket{\psi_S}-|\bra{\psi_S}{\dot{H}(\theta)}\ket{\psi_S}|^2\right]\\
    &=&4\Delta^2{\dot{H}(\theta)}.
\end{eqnarray}
Here $\Delta^2{\dot{H}(\theta)}$ denotes the variance of ${\dot{H}(\theta)}$ with respect to the initial state $\ket{\psi_S}$. Since $H(\theta)$ is Hermitian, ${\dot{H}(\theta)}$ is also Hermitian. Therefore, the set of normalized eigenvectors $\{\ket{i}\}$ of ${\dot{H}(\theta)}$, with $i=0,1,\ldots, d_S-1$, forms a complete orthonormal basis. Thus, one can write the initial state in terms of the basis vectors $\{\ket{i}\}$ as  
\begin{equation*}
\ket{\psi_S}=\sum_i{\omega_i}\ket{i}.
\end{equation*}
If $\{\epsilon_i\}$ denotes the set of eigenvalues of ${\dot{H}(\theta)}$ corresponding to the set of eigenvectors $\{\ket{i}\}$, then the QFI in terms of $\omega_i$ and $\epsilon_i$ can be written as
\begin{eqnarray}
\begin{split}
\mathcal{F}_{Q|\rho_S}^U(\theta)&=4\Bigg[\sum_i\mathcal{P}_i\epsilon^2_i-\Big(\sum_i\mathcal{P}_i\epsilon_i\Big)^2\Bigg]\\ 
&=4\Bigg[\sum_i|\omega_i|^2\epsilon^2_i-\Big(\sum_i|\omega_i|^2\epsilon_i\Big)^2\Bigg],
\end{split}
\end{eqnarray}
where $\mathcal{P}_i=|\omega_i|^2 \geq 0$ denotes the probability of $\ket{\psi_S}$ being in an eigenstate $\ket{i}$. To deduce the ultimate optimal QFI $\mathcal{F}_Q^U(\theta)$, one has to maximize $\mathcal{F}_{Q|\rho_S}^U(\theta)$ over all choices of the initial states $\ket{\psi_S}$, which is equivalent to maximizing over the choice of the set $\{\mathcal{P}_i\}$ such that $\sum_i\mathcal{P}_i=1$. Thus we have
\small
\begin{eqnarray}\label{p}
\mathcal{F}_{Q}^U(\theta)&=&\max_{\{\mathcal{P}_i\}\,|\,\sum_i\mathcal{P}_i=1}\Bigg[4\Big(\sum_i\mathcal{P}_i\epsilon^2_i-\Big(\sum_i\mathcal{P}_i\epsilon_i\Big)^2\Big)\Bigg].\quad
\end{eqnarray}
\normalsize
Now let us move on to the case of estimating the parameter $\theta$ using general quantum measurements. The discussion on general quantum measurement schemes in Sec.~\ref{premea} suggests that in this case, the probe state after encoding and before the measurement step may be correlated with an external auxiliary system $A$, of dimension $d_A$. In this context, during the encoding process, we can think of $U({\theta}) \otimes \mathbbm{I}_{d_A}$ as a higher-dimensional unitary acting on a composite system $\rho_{SA}$ of the system and auxiliary, where $\rho_{SA}$ may be correlated. Therefore, $\rho_{SA}$ can be regarded as an effective probe state undergoing the unitary encoding process $U \otimes \mathbbm{I}_{d_A}$. Consequently, in such a scenario, the generalized QFI corresponding to a given initial state $\tilde{\rho}_S=\Tr_A[\rho_{SA}]$ under the action of the unitary $U$ is equivalent to the QFI for the system–auxiliary joint state $\rho_{SA}$ under the action of the unitary $U \otimes \mathbbm{I}_{d_A}$. Thus, we have
$$ \mathcal{F}_{G|\tilde{\rho}_S}^U(\theta)=\mathcal{F}_{Q|\rho_{SA}}^{U\otimes \mathbb{I}_{d_{A}}}(\theta).$$
Correspondingly, using the convexity property of QFI, the ultimate generalized QFI maximized over all choices of input states of the probe can be written as  
$$ 
\mathcal{F}_{G}^U(\theta)=\mathcal{F}_{Q}^{U\otimes \mathbb{I}_{d_{A}}}(\theta)=  
\max_{\tilde{X} \in \tilde{\Lambda}^Q, \,\rho_{SA} \in \tilde{\zeta}_d} 
\Big(\mathcal{F}_{\tilde{X}|\rho_{SA}}^{U\otimes \mathbb{I}_{d_{A}}}(\theta)\Big), 
$$
where $\tilde{X}$ denotes a POVM on the composite system of the probe and the auxiliary, and $\tilde{\zeta}_d$ is the set of all pure states defined in the joint Hilbert space of the system and the auxiliary. Considering pure initial states of the form $\rho_{SA}=\ketbra{\psi_{SA}}{\psi_{SA}}$, we know from the discussion in Sec.~\ref{2} that the optimal POVM measurement strategy that yields the QFI is the one involving projective measurement in the SLD basis of $\rho_{SA}=\ketbra{\psi_{SA}}{\psi_{SA}}$. Thus, we have  
\begin{equation}\label{81}
     \mathcal{F}_{Q|\rho_{SA}}^{U\otimes \mathbb{I}_{d_{A}}}(\theta)
     =4\left[ \langle \dot{\psi}_{SA}(\theta) | \dot{\psi}_{SA}(\theta) \rangle 
     - \big|\langle \dot{\psi}_{SA}(\theta) | \psi_{SA}(\theta)\rangle\big|^2 \right],
\end{equation}
where
\begin{equation}
   \ket{\dot{\psi}_{SA}(\theta)}=-\iota\,\exp\!\big(-\iota H(\theta) \otimes \mathbb{I}_A\big)\,\dot{H}\otimes\mathbb{I}_A(\theta)\ket{\psi_{SA}}.
\end{equation}

The expressions $\langle\dot{\psi}_{SA}(\theta)|\dot{\psi}_{SA}(\theta)\rangle$ and $|\langle \dot{\psi}_{SA}(\theta) | \psi(\theta)_{SA} \rangle|^2$ can be written in terms of the operator $H(\theta)$ and its derivative ${\dot{H}(\theta)}$ as $\bra{\psi_{SA}}{\dot{H}(\theta)}^{2}\ket{\psi_{SA}}$ and $|\bra{\psi_{SA}}{\dot{H}(\theta)}\ket{\psi_{SA}}|^2$, respectively. Substituting this into the Eq.~\eqref{81}, we obtain
\begin{eqnarray}\label{9}
\mathcal{F}(\theta)_{G|\tilde{\rho_{S}}}^{U}\nonumber&=&4[\bra{\psi_{SA}}({\dot{H}^2(\theta)}\otimes \mathbbm{I}_{d_A})\ket{\psi_{SA}}\\\nonumber
 &-&|\bra{\psi_{SA}}{\dot{H}(\theta)}\otimes\mathbbm{I}_{d_A}\ket{\psi_{SA}}|^2]\\
    &=&4\Delta^2({\dot{H}(\theta)}\otimes\mathbbm{I}_{d_A}).
\end{eqnarray}
Here $\Delta^2({\dot{H}(\theta)}\otimes \mathbbm{I}_{d_A})$ denotes the variance of ${\dot{H}(\theta)}\otimes\mathbbm{I}_{d_A}$, with respect to the state $\ket{\psi_{SA}}$. Since ${\dot{H}(\theta)}$ is a Hermitian operator, the operator ${\dot{H}(\theta)} \otimes \mathbbm{I}_{d_A}$ is also Hermitian and shares the same set of eigenvalues $\{\epsilon_i\}$ as ${\dot{H}(\theta)}$, but each with $d_A$-fold of degeneracy. We can therefore expand the initial state $\ket{\psi_{PA}}$ in the orthonormal eigenbasis of ${\dot{H}(\theta)} \otimes \mathbbm{I}_{d_A}$, denoted by ${\ket{ij}}$, where $i=0,1,\ldots,d_S-1$ and $j=0,1,\ldots,d_A-1$. Thus we have $\ket{\psi_{SA}} = \sum_{i,j}{\omega_{ij}}\ket{ij}$. Note that each eigenvector $\ket{ij}$ have the same eigenvalue $\epsilon_i$, for all $j=0,1,\ldots d_A-1$. The quantity $\mathcal{P}_{ij}=|{\omega_{ij}}|^2$ denotes the probability of $\ket{\psi_{SA}}$ being in a particular eigenstate $\ket{ij}$  where $0\leq\mathcal{P}_{ij}\leq1$ and $\sum_{ij}\mathcal{P}_{ij}=1$. 
Thus the total probability of obtaining a particular eigenvalue $\epsilon_i$ is given as $\mathcal{P}_i=\sum_j\mathcal{P}_{ij}$. Therefore in terms of $\mathcal{P}_i$ the generalized QFI corresponding to the state $\rho_{SA}=\ketbra{\psi_{SA}}{\psi_{SA}}$ can be written as

\begin{eqnarray}
 \mathcal{F}(\theta)_{G|\tilde{\rho_{S}}}^{U}\nonumber(\theta)&=&4\Big[\sum_{i}\mathcal{P}_{i}\epsilon^2_{i}-(\mathcal{P}_{i}\epsilon_{i})^2\Big].
\end{eqnarray}

Now, the optimal generalized QFI maximized over all possible initial states $\ket{\psi_{SA}}$ and therefore over all choice of the set $\{\mathcal{P}_i\}$, satisfying the normalization condition $\sum_i\mathcal{P}_i=1$ is given as

\begin{eqnarray}\label{N}
\mathcal{F}_G^U(\theta)&=&\max_{\mathcal{P}_{i}|\sum_i\mathcal{P}_i=1}4\Big[\sum_{i}\mathcal{P}_{i}\epsilon^2_{i}-(\mathcal{P}_{i}\epsilon_{i})^2\Big].
\end{eqnarray}

Note that the above equation has exactly the same form as the optimal QFI corresponding to the POVM measurement setting for the unitary encoding $U$, as given in Eq.~\eqref{p}. In other words, both $\mathcal{F}_{Q|\rho_S}^U(\theta)$ and $\mathcal{F}_{G|\tilde{\rho}_S}^U(\theta)$ have the same functional form and are optimized over the same set of parameters $\{\mathcal{P}_i\}$. Consequently, one can conclude that for an arbitrary unitary encoding $U$, one always has 
$
\mathcal{F}_{Q}^U(\theta)=\mathcal{F}_{G}^U(\theta).
$ 
This completes the proof of Theorem~1.

Below, we present a remark that follows directly from Theorem~1.
\end{proof}
\textbf{Remark.} 
Theorem~1 suggests that, for unitary encoding, POVM operations suffice to obtain the optimal precision, even when the maximization is performed over all general quantum measurements and choice of input states. This highlights that NPOVM operations cannot increase the precision of parameter estimation in the case of unitary encoding. Moreover, NPOVM measurements are generally more costly than POVMs, since they require an initial correlation (a resource) between the probe and an external auxiliary. Therefore, for unitary encoding, one can avoid the additional cost of implementing NPOVMs and instead restrict to the less costly POVM measurement strategy while still achieving the optimal precision.


This completes our analysis for the case of arbitrary unitary encoding. In the next section, we consider open encoding and examine whether maximizing the QFI over general measurement settings and all possible input states can enhance the precision of parameter estimation in this case.

{\section{Positive vs non-positive measurements for open encoding}\label{4}}
{In the previous section, we analyzed the unitary encoding process in parameter estimation, where the probe remains ideally isolated from the environment and showed that POVM operation is enough to provide best precision for a suitable choice of initial probe state. In this section, we focus on the open encoding process, where the encoding invariably involves an interaction between the system and an external environment. This is modeled by the action of a global unitary on the system and the external followed by tracing out the environment.}

{Let $\rho_S$ and  $\rho_E$ denote the initial probe and external environment states acting on the Hilbert spaces $\mathcal{H}^S$ and  $\mathcal{H}^E$, respectively. The open encoding channel, $\Lambda_{\theta}$, is then a completely positive trace-preserving (CPTP) map that transforms $\rho_S$ into $\rho_S(\theta)$, in the following way}
\begin{equation*}
\Lambda_{\theta}(\rho_S) = \Tr_E\left[\mathcal{X}(\rho_S \otimes \rho_E)\mathcal{X}^\dagger\right] = \rho_S(\theta).
\end{equation*}
{Here the encoded parameter $\theta$ is a dimensionless quantity that is intrinsic to the global unitary $\mathcal{X}$. 


{In order to estimate $\theta$, one must perform a measurement on the encoded state, $\rho_S(\theta)$. As discussed in Sec.~\ref{premea} and Sec.~\ref{para}, such a measurement can either be restricted to POVMs or extended to more general quantum measurements. Accordingly, one obtains the optimal QFI, $\mathcal{F}_Q^O$, defined in Eq.~\eqref{ofi}, or the optimal generalized QFI, $\mathcal{F}_G^O$, defined in Eq.~\eqref{ogfi}, respectively.} {This naturally raises the question of whether, as in the case of unitary encoding, POVM measurements suffice to achieve the ultimate optimal precision for all types of open encoding, or whether there exist instances for which $\mathcal{F}_G^O > \mathcal{F}_Q^O$, implying that in those cases non-POVM schemes would yield better precision.}

{In this regard, considering a probe and environment with dimensions $d_S = 2$ and $d_E = 2$, respectively, we establish a sufficient condition, presented as a theorem in Sec.~\ref{case1}, that specifies when {POVM operations achieve the same level of precision as general operations.} In Sec. \ref{new_sec}, we provide few examples of open encoding channels for which POVM operations achieve the maximum precision. Furthermore, in Sec.~\ref{4Bx}, we identify scenarios in which NPOVM operations outperform standard POVMs in terms of optimal precision in estimation.}

{\subsection{When POVMs suffice 
in case of open encoding}\label{case1}}

{As depicted in the lower panels of Fig.~\ref{fig3x}, the open encoding strategy requires a global unitary operation $\mathcal{X}$, defined on the joint Hilbert space $\mathcal{H}^S \otimes \mathcal{H}^E$, to be implemented on the probe–environment composite system, followed by tracing out the environment. Here we consider both the probe and the environment to be two-dimensional systems (qubits). The auxiliary system, however, is described by a Hilbert-space, $\mathcal{H}_A$, of an arbitrary dimension, $d_A$. Since in $\mathbb{C}^2 \otimes \mathbb{C}^2$ any unitary operation can be decomposed in terms of tensor products of single-qubit operators $\sigma_i$ ($i=0,1,2,3$), with $\sigma_0 = \mathbb{I}_2$ and $\sigma_i$ for $i=1,2,3$ being the Pauli matrices, one can write
\begin{equation*}
\mathcal{X}=\sum_{ij} h(\theta)_{ij}\, \sigma_i \otimes \sigma_j ,
\end{equation*}
where the parameter $\theta$ to be estimated is an intrinsic part of the unitary $\mathcal{X}$, and each coefficient in the expansion of $\mathcal{X}$ is an implicit function of $\theta$.} 

{The nature of the encoding, and hence the form of $\mathcal{X}$ and the initial environment state $\rho_E$, is taken to be the same for both decoding strategies (POVM or general quantum operations). However, as shown in Fig.~\ref{fig3x}, the only difference between the two strategies lies in the initial system-auxiliary state and final measurement, i.e., in the POVM case, the initial system-auxiliary state, $\rho_S\otimes\rho_A$, must be product whereas in case of general measurement it can be any state, $\rho_{SE}$ of dimensionn $d_S\otimes d_A$. Additionally, the final measurement basis can also be different in the two strategies defined by $U_M$ and $U_M'$ for general and positive measurements respectively, as depicted in Fig.~\ref{fig3x}. Within this setup, we present below the theorem that provides a sufficient condition for the precision attainable by a general quantum operations to be exactly reachable by a suitably chosen POVM for a given open encoding.

\begin{theorem}
{For an open encoding, by acting a unitary on the probe-environment state $\rho_S\otimes\rho_E$, the POVM-decoding strategy with a particular measurement setting and a probe-auxiliary initial state of the form 
$
\rho_S \otimes \rho_A = \tilde{U}_1 \otimes \tilde{U}_2 \ketbra{\alpha_S \alpha_A}{\alpha_S \alpha_A} \tilde{U}_1^\dagger \otimes \tilde{U}_2^\dagger,
$
where $\ket{\alpha_S \alpha_A}$ is any fixed pure product state of the Hilbert space $\mathcal{H}^S \otimes \mathcal{H}^A$, yields the same precision as that obtained by the general quantum measurement strategy with a fixed choice of measurement setting and an initial (possibly correlated) state 
$\rho_{SA} = U \ketbra{\alpha_S \alpha_A}{\alpha_S \alpha_A} U^\dagger$,
if there exist an $i$ for which
$$
   U \, (\mathbbm{I} \otimes \tilde{U}_2 \sigma_i \sigma_k \tilde{U}_2^{\dagger}) 
   = (\mathbbm{I} \otimes \sigma_i \sigma_k) \, U, 
   \qquad \forall k.
$$
Here, $\tilde{U}_1$ $\tilde{U}_2$, and $U$ are unitary operators which act on the probe, auxiliary, and probe–auxiliary system, respectively.}
\end{theorem}

\begin{proof}
{Let us first consider the case of general measurements. The initial system-auxiliary state in this scenario is $\rho_{SA}=U \ketbra{\alpha_S \alpha_A}{\alpha_S \alpha_A} U^\dagger$ (as specified in the statement of the theorem). 
Another environment is introduced for the open encoding which is initially in a state $\rho_E$ of dimension $d_E=2$ and can be expressed in its spectral decomposition as $\rho_E=\sum_{i=1}^2r_i\ket{\mathcal{E}_i}\bra{\mathcal{E}_i}$. Hence, the composite probe-auxiliary-environment state is $\rho_{SA}\otimes\rho_E=\sum_{j}r_{j}U\ket{\alpha_S \alpha_A}\bra{\alpha_S \alpha_A}U^\dagger\otimes\ket{\mathcal{E}_j}\bra{\mathcal{E}_j}$. Since the global unitary $\mathcal{X}$ used to encode a parameter $\theta$, acts on $2\times 2$ systems, it can be decomposed in terms of Pauli matrices and two-dimensional identity matrix as, $\mathcal{X}=\sum_{ij}h(\theta)_{ij}\sigma_i \otimes \sigma_j$. Hence the final system-auxiliary-environment state, after encoding is given by 
\begin{widetext}
{
\begin{eqnarray}\nonumber
\rho{'}_{SAE}&=&\nonumber
\sum_{i,j,k,l}\big[h(\theta)_{ij} h(\theta)^*_{kl}r_1 \mathbbm{I}\otimes \sigma_iU\ket{\alpha_S \alpha_A}\bra{\alpha_S \alpha_A}U^{\dag} \mathbbm{I}\otimes \sigma_k\sigma_j\ket{\mathcal{E}_1}\bra{\mathcal{E}_1}\sigma_l\\\nonumber
&+&h(\theta)_{ij} h(\theta)^*_{kl}r_2 \mathbbm{I}\otimes \sigma_iU\ket{\alpha_S \alpha_A}\bra{\alpha_S \alpha_A}U^{\dag}\mathbbm{I}\otimes \sigma_k\sigma_j\ket{\mathcal{E}_2}\bra{\mathcal{E}_1}\sigma_l\big].
\end{eqnarray}}
\end{widetext}
{The final state of probe-auxiliary system after tracing out the external environment is}
{\begin{eqnarray}\nonumber
\rho{'}_{SA}&=&r_1\sum_{ik}A_{ik}\mathbbm{I}\otimes \sigma_iU\ket{\alpha_S \alpha_A}\bra{\alpha_S \alpha_A}U^{\dag}\mathbbm{I}\otimes \sigma_k\\ \nonumber
&+&r_2\sum_{ik}B_{ik}\mathbbm{I}\otimes \sigma_iU\ket{\alpha_S \alpha_A}\bra{\alpha_S \alpha_A}U^{\dag}\mathbbm{I}\otimes \sigma_k \nonumber.
\end{eqnarray}}
To perform the measurement on in an arbitrary but fixed directionn one can first rotate the state, $\rho_{SA}$, using an apropriate unitary, say $U_M$, and project the final state on the computational basis. Hence the system-auxiliary state, just before its projection on the computational basis will be 
{\begin{eqnarray}\nonumber
&&\rho{''}_{SA}\nonumber\\&=&r_1\sum_{ik}A_{ik}U_M\mathbbm{I}\otimes \sigma_iU\ket{\alpha_S \alpha_A}\bra{\alpha_S \alpha_A}U^\dagger\mathbbm{I}\otimes \sigma_k U_M^{\dag}\nonumber\\ \nonumber
&+&r_2\sum_{ik}B_{ik}U_M\mathbbm{I}\otimes \sigma_iU \ket{\alpha_S \alpha_A}\bra{\alpha_S \alpha_A}U^\dagger\mathbbm{I}\otimes \sigma_kU_M^{\dag} \nonumber.
\end{eqnarray}.}
{On the other hand, to implement the POVM, it is sufficient to initially consider $U=\tilde{U_1}\otimes\tilde{U_2}$, i.e., a local unitary. The final projective measurement on the probe-auxiliary state can again be performed by rotating the encoded prob-auxiliary state in a suitable direction, say $U'_M$, and projecting it on the computational basis. Thus, in this case the state after the rotation and before the projection would be}
\begin{eqnarray}\nonumber
\rho{'''}_{SA}&=&r_1\sum_{ik}A_{ik}\mathcal{A}_{ik}
+r_2\sum_{ik}B_{ik}\mathcal{B}_{ik}.
\end{eqnarray}
Here the operator $\mathcal{A}_{ik}$ and $\mathcal{B}_{ik}$ are given by $\mathcal{A}_{ik}=U'_M\mathbbm{I}\otimes \sigma_i(\tilde{U_1}\otimes\tilde{U_2})\ket{\alpha_S \alpha_A}\bra{\alpha_S \alpha_A}(\tilde{U}_1^{\dag}\otimes\tilde{U}_2^{\dag})\mathbbm{I}\otimes \sigma_k U_M'^{\dag}$ and $\mathcal{B}_{ik}=U'_M\mathbbm{I}\otimes \sigma_i(\tilde{U_1}\otimes\tilde{U_2}) \ket{\alpha_S \alpha_A}\bra{\alpha_S \alpha_A}(\tilde{U}_1^{\dag}\otimes\tilde{U}_2^{\dag})\mathbbm{I}\otimes \sigma_kU_M'^{\dag}$ respectively.
{A sufficient condition for the Fisher information corresponding to POVM and non-POVM to be equal is the states, $\rho''_{PA}$ and $\rho'''_{PA}$, are itself equal. In order to make the states equal, we can first make few terms of them equal by tuning $U_M'$ so that it satisfy the following condition}
{\begin{equation}\label{14}
 U_M(\mathbbm{I}\otimes \sigma_i)U=U'_M(\mathbbm{I}\otimes \sigma_i)(\tilde{U_1}\otimes\tilde{U_2}),
 \end{equation}
 for a fixed $i$. This makes $U_M'$ fixed as 
 \begin{equation}
 U'_M=U_M(\mathbbm{I}\otimes \sigma_i)U(\tilde{U_1}\otimes\tilde{U_2})(\mathbbm{I}\otimes \sigma_i).\label{myeq1}   
\end{equation}}
{One can easily check that the states $\rho''_{PA}$ and $\rho'''_{PA}$ can entirely be equal if the following condition is satisfied}
\begin{equation}\label{15}
    {U'_M}\left(\mathbbm{I}\otimes \sigma_k\right)(\tilde{U_1}\otimes\tilde{U_2})=U_M\left(\mathbbm{I}\otimes \sigma_k\right)U,~\forall k.
\end{equation}
{Using Eq.~\eqref{myeq1} and Eq.~\eqref{15}, we get the ultimate sufficient condition which is given by}
\begin{equation}\label{suff1}
    U\left(\mathbbm{I}\otimes \tilde{U}_2\sigma_i\sigma_k\tilde{U}_2^{\dag}\right)=\left(\mathbbm{I}\otimes \sigma_i\sigma_k\right)U, ~\forall k.
\end{equation}
{This is the sufficient condition of equality of Fisher information achieved by POVM and any general measurement.}}
\end{proof}
\begin{remark}
{Note that though we have provided the sufficient condition considering pure initial states (aprior to encoding) the same sufficient condition will also be valid if the initial states, $\rho_S\otimes\rho_A$ and $\rho_{SA}$, are mixed but have the same set of eigenvalues.
}
\end{remark}
\begin{remark}
{ If $U$ is the optimal unitary corresponding to the general operation, then Eq.~\eqref{suff1} provides the sufficient condition that POVM operation is enough for reaching best possible precision among all general measurements, i.e., the QFI for general measurements. }
\end{remark}
\begin{remark}
{Since we decomposed the global unitary involved in encoding in terms of tensor products of Pauli and identity matrices, i.e., $\{\sigma_i\otimes\sigma_k\}_{ik}$, the final sufficient condition depends on those matrices. If instead of $\{\sigma_i\otimes\sigma_k\}_{ik}$, the decomposition of the global unitary is available in terms of a set of any other local operators, say $\{\tau_i\otimes\omega_k\}_{ik}$, similar sufficient conditions can still be found following the same mathematical logic. The condition in such case will be $    U\mathbbm{I}\otimes \tilde{U}_2\omega_i\omega_k\tilde{U}_2^{\dag}=\mathbbm{I}\otimes \omega_i\omega_k U, \forall k.$ Hence the condition remains the same even when one moves to higher dimensional auxiliary and prob systems instead of being restricted to qubits and replaces $\{\sigma_i\otimes\sigma_k\}_{ik}$ with the higher dimensional set of local operators that are involved in the decomposition of the global unitary.
}
\end{remark}

\begin{figure}
		\centering

            \includegraphics[width=8.0cm]{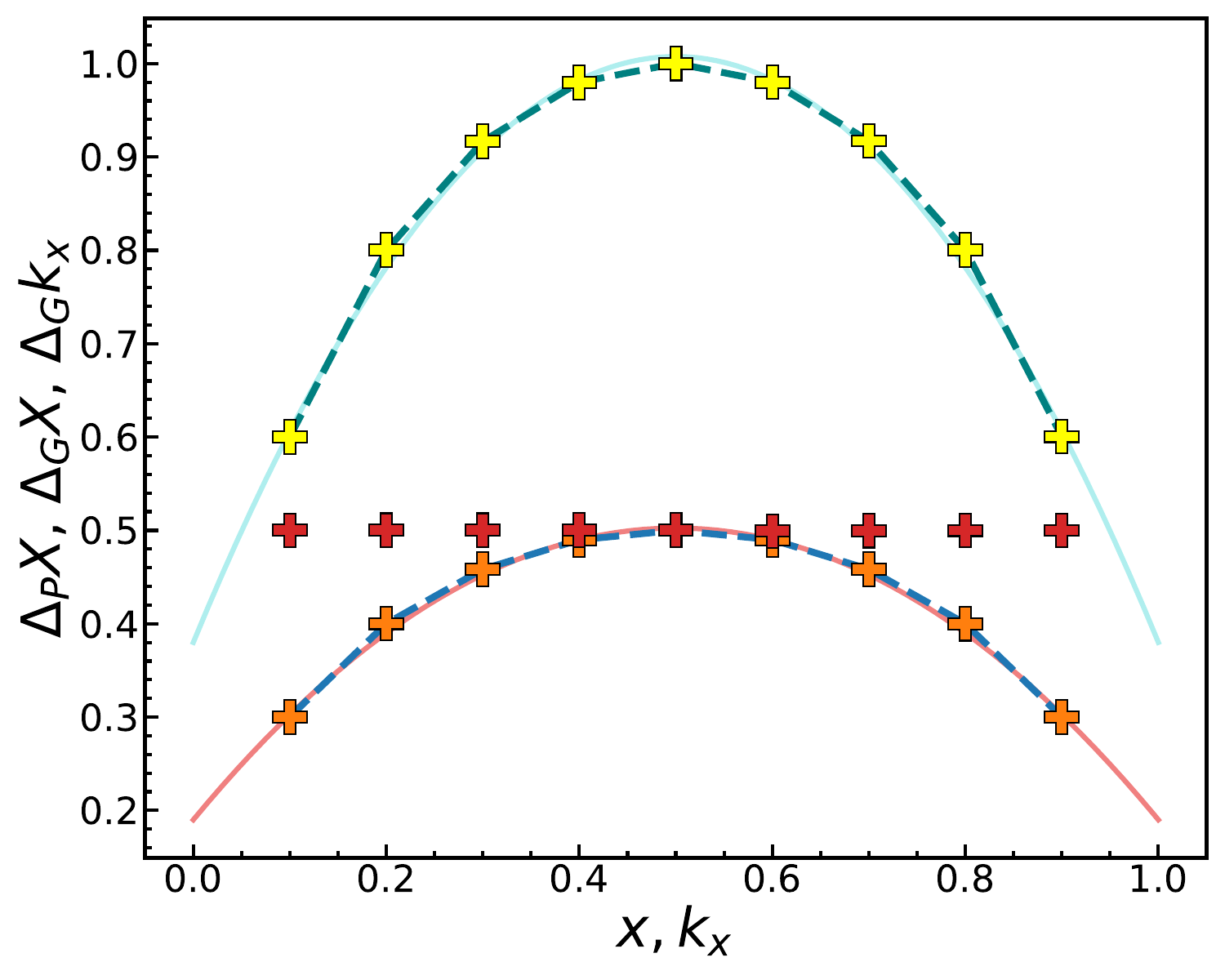}
\caption{\textbf {POVM providing the best precision.} Here the minimum error in the parameter estimation problem by optimal general measurement is shown by yellow, orange, and red pulse markers corresponding to dephasing, amplitude damping, and $XX$ interaction. On the other hand, the green and blue dashed lines represent the minimum estimation error achievable using POVM measurements corresponding to dephasing and amplitude damping noise. The fitted curves corresponding to dephasing and amplitude damping channels are given by the cyan and orange curves. }
		\label{fig3xp} 
		\end{figure} 
\subsection{Examples of open encoding for which POVMs are optimum for estimation}
\label{new_sec}
In this part, we analyze a few specific examples of open encoding and compare the precision achieved through POVM and general operations. First let us consider two open encoding channels, viz. amplitude damping and dephasing. Action of a CPTP map, $\mathcal{M}$, on a state, $\rho$, can be written in terms of Kraus operators, $\{K_i\}_i$, as $\mathcal{M}(\rho)=\sum_i K_i \rho K_i^\dagger$. The operations of amplitude damping and dephasing channels are not exceptions. In particular, for amplitude damping channel, the Kraus operators can be taken as $K_0=\ket{1}_S\bra{1}+\sqrt{1-k_a}\ket{0}_S\bra{0}$ and $K_0=\sqrt{k_a}\ket{1}_S\bra{0}$. 
The Kraus operators corresponding to dephasing channel can be considered to be $K_0=\sqrt{k_d}\mathbbm{I}_{2}$, $K_1=\sqrt{k_d}\ket{0}_S\bra{1}$ and $K_2=\sqrt{k_d}\ket{1}_S\bra{0}$. Here, $k_a$ and $k_d$ are the noise strengths of amplitude damping and dephasing channels respectively which we aim to estimate.
Along the vertical axis of Fig.~\ref{fig3xp}, we plot the maximum precisions, $\Delta_PX$ and $\Delta_\mathcal{G}X$, in estimating the noise strengths encoded through amplitude damping and dephasing channels using the optimal positive measurement and optimal general measurements, respectively. The horizontal axis represents the actual value of the estimated parameter, $X$. Here $X$ denotes the two parameters under consideration, i.e., $X\in\{k_a,k_d\}$. The blue and green dashed lines represent $\Delta_PX$ for estimation of the noise strength $k_a$ and $k_d$, respectively, using POVM operations. The behaviors of $\Delta_\mathcal{G}X$ are illustrated using orange and yellow stars for estimation of, respectively, $k_a$ and $k_d$. One can notice from the plots, $\Delta_PX\approx \Delta_\mathcal{G}X$, for all considered parameter values. Hence in these cases, the POVMs are sufficient to gain optimum precision and one need not to employ resource-intensive NPOVM measurements. As can be noticed from the plots, the behavior of $\Delta_{P,\mathcal{G}}X$ with respect to $X$ is parabolic. By fitting the curves with the equations  $\Delta_{P,\mathcal{G}}k_a=A_1k_a^2-A_1k_a+A_2$ and $\Delta_{P,\mathcal{G}}k_d=D_1k_d^2-D_1k_d+D_2$ we find the exact form of the parabolas, which are described by the coefficients $A_1=-1.25$, $A_2=0.192$, $D_1=-2.51$ and $D_2=0.384$. The fitted curves corresponding to  dephasing and amplitude damping noise are represented by the cyan and orange curves, respectively.} The least squares error of the fitting curves corresponding to amplitude damping and dephasing noise are $52\%$ and $59\%$, respectively.
\begin{figure*}
		\centering
			\includegraphics[width=5.4cm]{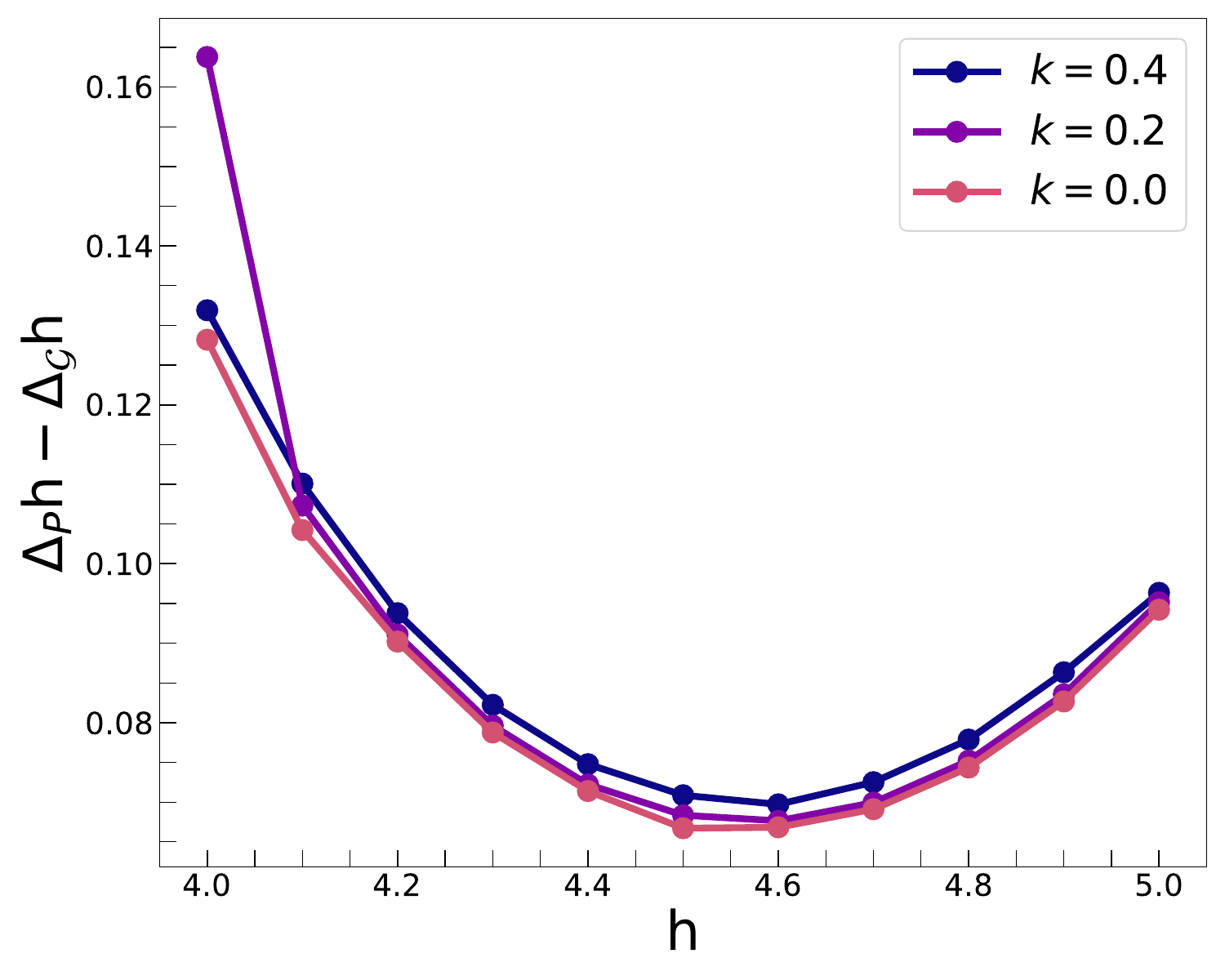}
   \includegraphics[width=5.4cm]{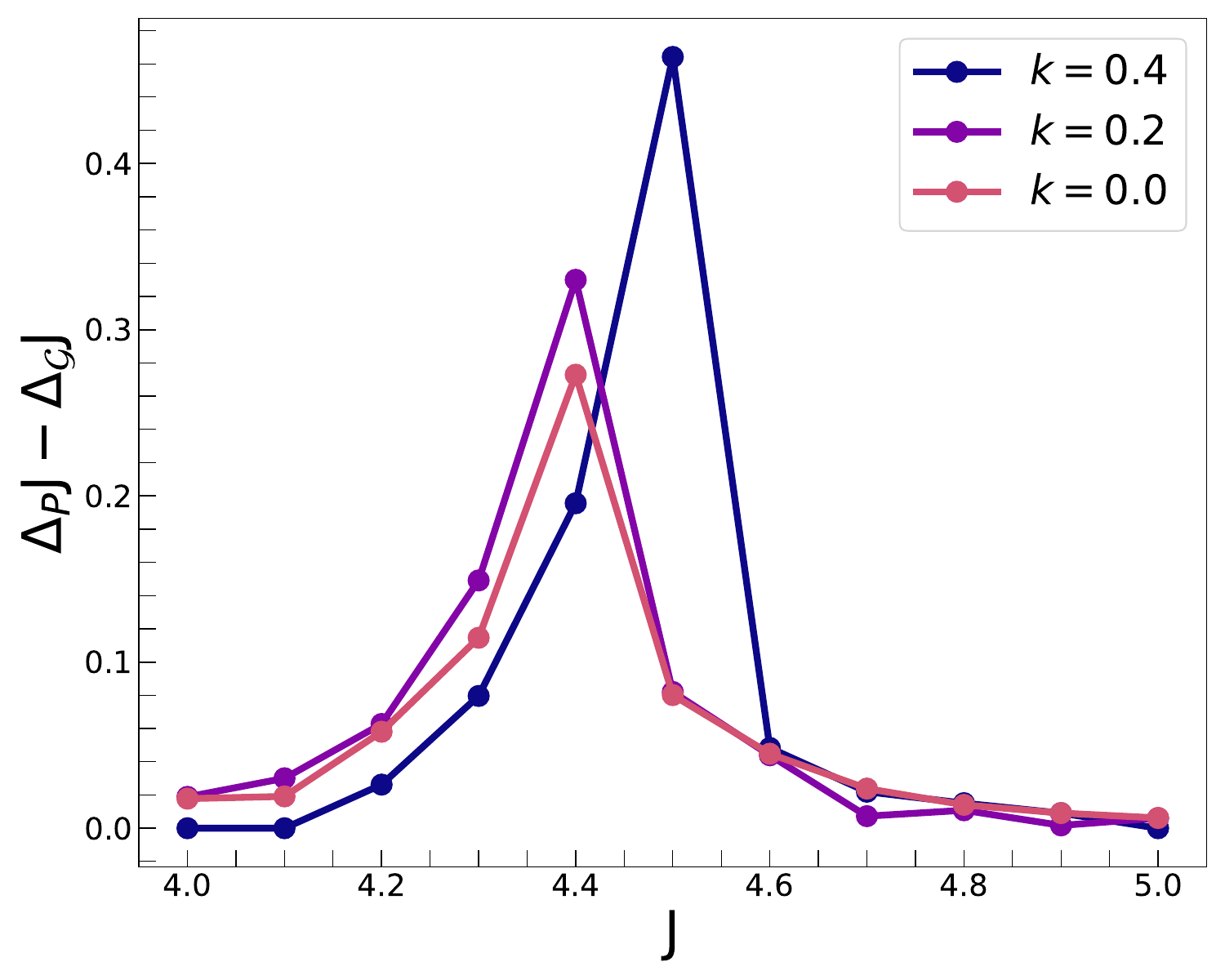}
    \includegraphics[width=5.4cm]{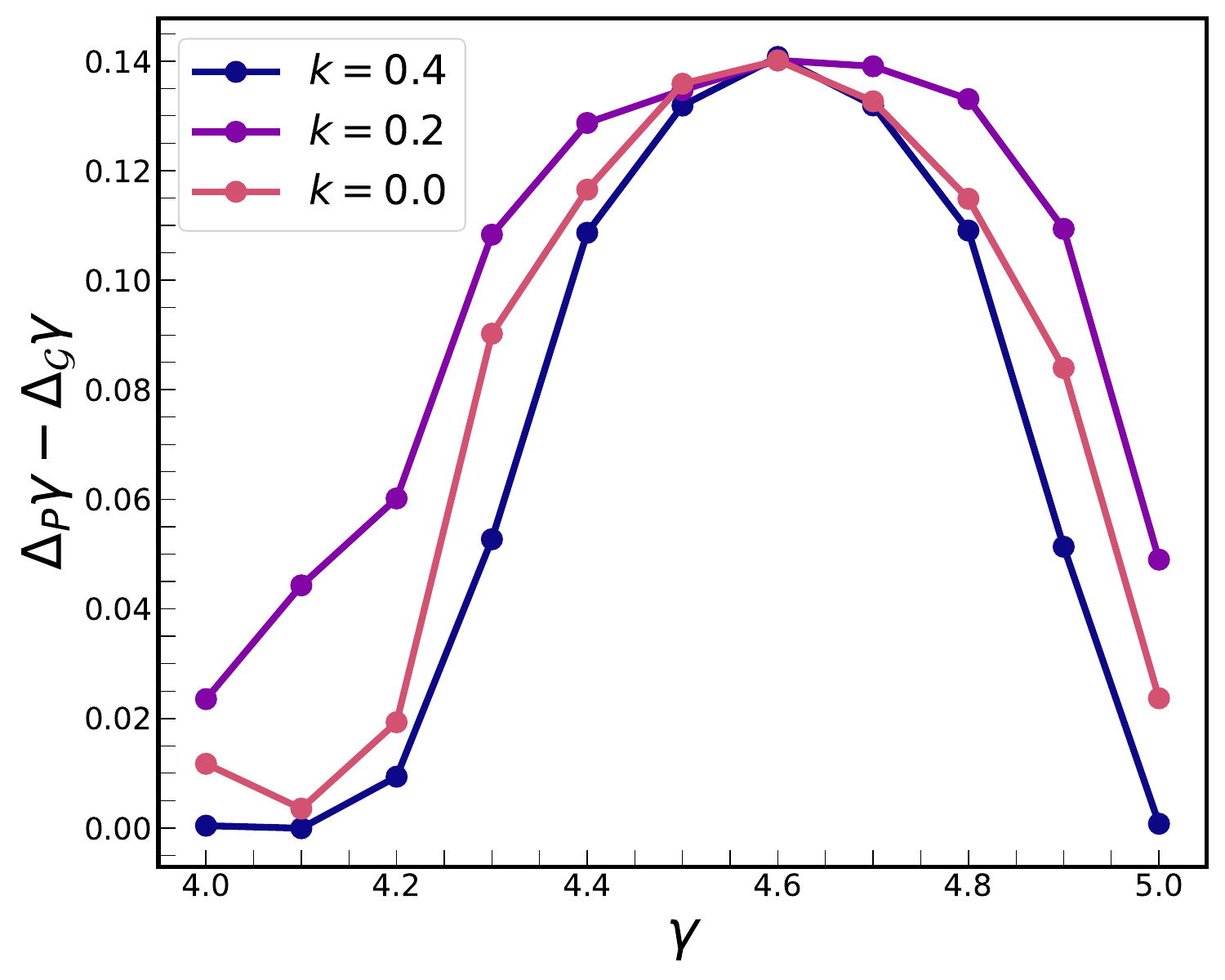}
\caption{{\textbf{Enhancement of precision using non-positive measurements.}} Here we have plotted the differences of minimal error achived by POVM and general measurement operation Vs. the estimated parameter. The left handed figure is corresponds to the field parameter $h$, where $J=\hbar/t$ and $\gamma=1$. In the middle figure the estimated parameter is $J$, where $h=\hbar/t$ and $\gamma=1$. At last the considered parameter is $\gamma$, where, we have taken $J=\hbar/t$ and $h=\hbar/t$. Here, we observe $\Delta_{P}h - \Delta_{\mathcal{G}}h$  decreases as $h$ increases and reaches a minimum value at $h=4.55\hbar/t$, and then begins to increase again. In the middle sub-figure, in the $\Delta_{P}J - \Delta_{\mathcal{G}}J$ vs. $J$ plot, $\Delta_{P}J - \Delta_{N}J$ increases with the parameter $J$, reaches a maximum at an intermediate value of $J$, and then begins to decrease. However, the peak corresponding to $J$ varies depending on the value of the initial environment state parameter $k$. Finally, in the rightmost figure, we have plotted $\Delta_{P}\gamma - \Delta_{\mathcal{G}}\gamma$ versus $\gamma$, which also shows sufficient enhancement of precision. } 
		\label{fig3}
		\end{figure*}

Let us now move to another example. Here, for open encoding we bring the external environment and switch on an XX interaction between the environment and the prob for a particular amount of time, say $t$. The interaction is governed by the Hamiltonian $H=\theta(\sigma_X \otimes \sigma_X)$, where $\theta$ is the strength of the interaction. The evolution can be described by the global unitary operator, $U_G=\exp(-\iota\theta(\sigma_x\otimes\sigma_x)t/\hbar)=\exp(-(k_x\sigma_x\otimes\sigma_x))$, where $k_x=t\theta/\hbar$. We want to estimate the parameter $k_x$, encoded in the prob through this interaction and analyze the effectiveness of POVMs in the estimation. In this regard, we first numerically find the optimum fisher information using general measurements and plot that in the vertical axis of Fig. \ref{fig3xp} with respect to $k_x$, presented in the horizontal axis, using red plus markers. The general measurement performed on the prob-auxiliary system can be realized by the application of a rotation, $U_M$, on them and then projection on the computational basis. We numerical find the optimum measurement can always be implemented through the application of the following global unitary 
\begin{equation}
U_M=\begin{bmatrix}
1 & 0 & 0 & 1\\
0 & 1 & 1 & 0\\
0 & 1 & -1 & 0\\
1&0 & 0 &-1\\
\end{bmatrix}.
\end{equation}
To check if the precision can be reached using POVMs, we make use of the theorem discussed in the previous subsection. In this regard, we decompose the unitary $U_G$ in the Pauli basis as $U_G=\sin(\theta)\mathbbm{I}_2\otimes I_2+\cos(\theta)\sigma_x \otimes \sigma_x$. Considering $\sigma_i=I_2$, $\sigma_k=\sigma_x$, $\tilde{U}_2=I_2$, and the given $U_M$, it can be easily checked that the sufficient condition provided by the theorem is being satisfied. Hence, in this case also, there exist a suitable initial state and measurement operators using which the positive measurements can provide the optimal precision.
{\subsection{NPOVM decoding can be better than POVMs}\label{4Bx}}
In Sec.~\ref{case1}, we provided the sufficient condition for the reachability of the optimum precision by POVM for a given open encoding channel, where the optimization is performed over the set of general measurements. We also presented particular examples where the optimal precision is shown to be attainable by POVM. 

In this section, we present instances when NPOVM measurement schemes can yield more precision than the traditional optimal precision achieved by restricting to only POVM measurement strategies. In this regard, we consider the initial probe system as a qubit, which is initially in the state $\rho_S$ acting on the Hilbert space $\mathcal{H}^S$. To encode a parameter on the probe system through open encoding, we consider a qubit environment that is initially available in the state $\rho_E=\frac{1+k}{2}\ket{0}\bra{0}+\frac{1-k}{2}\ket{1}\bra{1}$ acting on the Hilbert space $\mathcal{H}^E$. {Then a global unitary $U_{SA}=e^{-\iota H_{SE}t/\hbar}$ corresponding to the Hilbert space $\mathcal{H}^S\otimes \mathcal{H}^E$ acts on the probe and environment, where $H_{SE}$ is a two-qubit XY model~\cite{XY}, given as
\begin{eqnarray}\label{xop1}\nonumber
H_{SE}&=&h(\sigma_z\otimes\mathbbm{I}_2+ \mathbbm{I}_2\otimes\sigma_z)\\&+&\frac{J}{2}((1+\gamma)\sigma_x\otimes\sigma_x+(1-\gamma)\sigma_y\otimes\sigma_y).
\end{eqnarray}
After the global evolution governed by the unitary $U_{SE}$, the environment is discarded. Therefore the final encoded state is given by $\rho^e_S=\Tr_E[U_{SE} (\rho_S \otimes \rho_E) U^{\dag}_{SE}]$. {In the next step to perform the measurement, we bring an auxiliary qubit system and perform a projective measurement on the composite initial state of the encoded probe and auxiliary system. For POVM, the auxiliary system is kept in a product state with the encoded probe system before performing the measurement, whereas no such restriction is imposed in the case of the general measurement scheme. Now to maximize the precision by both POVM and general measurement, we maximize it over the measurement setting and the initial probe state. POVM is optimized by optimizing over the initial auxiliary state and the projective measurement basis that acts on the composite system of probe and auxiliary system. On the other hand, general measurement is optimization over the auxiliary and the entanglement contained between the probe and the projective measurement basis that acts on the joint system of the probe and the auxiliary system. 


The Hamiltonian in Eq.~\eqref{xop1} has three parameters: the field strength ($h$), interaction strength ($J$), and the anisotropic parameter ($\gamma$). In Fig.~\ref{fig3}, we plot the difference between the minimum error in estimation achieved by positive and the general measurements {corresponding to the parameters $h, J,$ and $\gamma$ vs. the corresponding parameters. We perform the analysis for three different values of $k$, i.e., $k=0.0$, $k=0.2$, and $k=0.4$. The difference between the minimum error achieved by POVM and general measurement reflects the difference of maximum precision achieved by POVM and general measurement. In the first situation, we aim to estimate the parameter $h$, where we set $\gamma=1$ and $J=\hbar/t$, as shown in the leftmost panel of Fig.~\ref{fig3}. Here, we observe that the gap between the minimum error of estimating $h$ using POVM and general operations, i.e., $\Delta_{P}h - \Delta_{\mathcal{G}}h$, decreases as $h$ increases. It reaches a minimum value at $h=4.55\hbar/t$, and then begins to increase again, for all considered values of $k$. However, throughout the given range of $h$, i.e., $[4.0,5.0]$ in $\hbar/t$ unit, $\Delta_{P}h - \Delta_{\mathcal{G}}h>0$. It implies that the NPOVM, which lies outside the set of all POVMs, plays the role of reducing the error. On the other hand, in the middle panel of the figure, we plot $\Delta_{P}J - \Delta_{\mathcal{G}}J$ with respect to $J$. In this case, we observe $\Delta_{P}J - \Delta_{\mathcal{G}}J$ increases with increasing $J$, reaching a maximum at an intermediate value of $J$, and then begins to decrease. However, the $J$ values at which the difference reaches the peak vary with $k$. Finally, in the rightmost panel, we plot $\Delta_{P}\gamma - \Delta_{\mathcal{G}}\gamma$ with respect to $\gamma$, which again shows a reduction of error in estimation of $\gamma$ by NPOVM operations. Here it is interesting to note that for the leftmost and middle panels corresponding to the parameters $h$ and $J$, as we increase $k$ from 0 to 0.4 with a 0.2 interval, the enhancement in precision through NPOVMs increases. These observations indicate the usefulness of NPOVM operations over POVM operations in quantum metrology.
\section{conclusion}
\label{5}
The potential advantages of non-positive quantum measurements have been demonstrated in various areas of quantum technologies. Inspired by these advances, in this study we turned our attention to quantum metrology and asked, can non-positive measurements also enhance precision in parameter estimation? To date, quantum metrology has relied exclusively on POVM-based measurement strategies. This naturally motivated the broader question of whether extending the framework to include general measurements, embracing both positive and non-positive measurements, can unlock new metrological advantages. Addressing this question not only enriches the theoretical landscape of quantum measurements but also positions NPOVM operations as a promising resource for advancing quantum metrology, thereby strengthening their role in the wider domain of quantum technologies.

To this end, we examined both unitary and open encoding strategies within the quantum metrological setting. For arbitrary unitary encodings, we proved analytically that POVMs alone are sufficient to attain the ultimate precision, even when compared against the more general class of measurements that include NPOVMs. Thus, in this setting, NPOVMs offer no metrological advantage.

We demonstrated a different situation in the open-encoding regime. In this case, two possibilities arose: in some instances, POVMs already reached the optimal precision achievable by any general measurement; in others, non-positive measurements proved strictly more powerful. We identified a sufficient condition that characterized when POVMs were guaranteed to be optimal and confirmed this result numerically through representative examples, including parameter estimation tasks involving bit-flip and dephasing channels, as well as $XX$-type interactions between the environment and an auxiliary system.

On the other hand, we depict some instances where NPOVMs demonstrated a clear advantage. This was illustrated through the transverse-field XY (TXY) model, where we studied the estimation of the field strength, interaction strength, and anisotropy parameter. In each of these scenarios, NPOVM-based measurements outperformed their POVM counterparts, highlighting their potential as a resource for enhanced quantum metrology.
\acknowledgements
KS acknowledges support from the project MadQ-CM (Madrid Quantum de la Comunidad de Madrid) funded by the European Union (NextGenerationEU, PRTRC17.I1) and by the Comunidad de Madrid (Programa de Acciones Complementarias).
\bibliography{metrology}
\end{document}